\documentclass[10pt,emptycopyrightspace]{ewsn-proc}

\usepackage{graphicx}
\usepackage{balance}
\usepackage{comment}
\usepackage{xcolor}
\usepackage{cite}
\usepackage{hyperref}
\usepackage{amsmath}
\usepackage{amssymb}
\usepackage{subfigure}

\newcommand{\dif}{\mathrm{d}} 
\newtheorem{thm}{Theorem}
\newtheorem{prop}[thm]{Proposition}

\numberofauthors{1}

\author{
\alignauthor Chaojie Gu $\qquad$ Linshan Jiang $\qquad$ Rui Tan \\
    \affaddr{School of Computer Science and Engineering}\\
    \affaddr{Nanyang Techonological University, Singapore}\\
    \email{gucj@ntu.edu.sg, linshan001@e.ntu.edu.sg, tanrui@ntu.edu.sg}
}

\title{LoRa-Based Localization: Opportunities and Challenges}

\begin{document}
\maketitle

\begin{abstract}
Low-power wide-area network (LPWAN) technologies featuring long-range communication capability and low power consumption will be important for forming the Internet of Things (IoT) consisting of many geographically distributed objects. Among various appearing LPWAN technologies, LoRa has received the most research attention due to its open specifications and gateway infrastructures unlike the closed designs and/or managed gateway infrastructures of other LPWAN technologies. While existing studies on LoRa has focused on network connectivity and performance, accurate positioning of LoRa end devices is still largely an open issue. In this paper, we discuss and analyze the physical layer features of LoRa that are relevant to localization. Our discussions and analysis illustrate the opportunities and challenges in implementing LoRa-based localization.
\end{abstract}

%
%

%

\section{Introduction}

Low-power wide-area networks (LPWANs) are an emerging wireless platform that aims to sustain power-constrained end devices (e.g., those based on batteries or energy harvesting) to operate for years while communicating at low data rates to gateways several kilometers away. LPWAN technologies will largely increase the degree of connectivity of Internet of Things (IoT) and enable deep penetration of IoT objects into the urban territories. Fig.~\ref{fig:distance_energy} illustrates the comparisons among various wireless technologies in terms of radio power consumption and communication ranges. From the figure, LPWANs (e.g., LoRaWAN \cite{lora-alliance}, Sigfox \cite{sigfox}, Weightless-P \cite{weightless}, and NB-IoT \cite{nb-iot}) form an important pole in the spectrum of radio power consumption versus communication range.

\begin{figure}
  \includegraphics[width=\columnwidth]{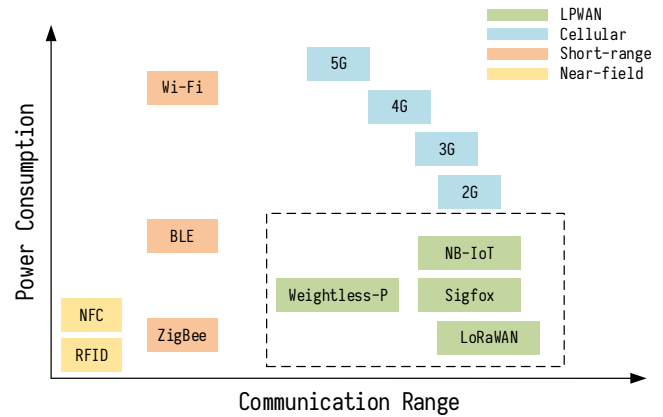}
    \caption{Power consumption versus communication range for various radios.}
    \label{fig:distance_energy}
\end{figure}

Given LPWAN signals' capability of propagating over long distances and penetrating civil infrastructures, exploiting LPWAN signals for localizing IoT objects over long distances and/or in indoor environments has attracted increasing research interests. In this paper, we discuss the opportunities and challenges of LoRa-based localization. LoRa is a physical layer technique that uses a Chirp Spread Spectrum (CSS) modulation, whereas LoRaWAN is an open data link layer specification based on LoRa. Our focus of studying LoRa for localization is due to its use of license-free ISM band (e.g., EU868 MHz and US916 MHz), low cost (US\$15 per unit \cite{low-cost}), and scalability to support many IoT objects. Compared with other wireless-based solutions, the LoRa-based localization will offer the following advantages:
\begin{itemize}
\item If an infrastructure-based localization approach is considered, the LoRa infrastructures will be much simpler than other wireless infrastructures. Existing approaches based on short-range radios, e.g., Wi-Fi, Bluetooth, and ultra-wideband (UWB), often require dense deployment of access points (APs) and beacon nodes, leading to high installation and maintenance costs. In contrast, due to LoRa's long-range communication capability, a small number of LoRa gateways can cover a geographic area (e.g., a campus) or a civil construction (e.g., a building) and act as anchors to support the localization of LoRa end devices. Thus, the deployment of the infrastructure will not incur significant overhead.
\item Due to LoRa's low power consumption and good scalability, the LoRa-based localization solutions can be applied to massive objects (e.g., goods in a warehouse). High-power radios (e.g., Wi-Fi and cellular) are ill-suited for such scenarios, since they will drain the battery energy quickly. The RFID and NFC solutions, which rely on the reader's scanning, cannot achieve real-time object localization.\vspace{0.5em}
\end{itemize}

Despite the above desirable features of LoRa-based localization, we also face the following key technical challenges:
\begin{itemize}
\item We aim to develop solutions that use customized gateways to localize off-the-shelf LoRa end devices without any hardware customization and retrofitting. As most commercial LoRa end devices do not provide access to the LoRa physical layer, it is challenging and even impossible to obtain low-level information that is often critical to localization.
\item After penetrating barriers, the LoRa signals often go below the noise floor. Although the LoRa's CSS modulation is designed to deal with significant signal attenuation, the localization approaches based on received signal strength (RSS) and RSS indicator (RSSI) will be vulnerable to low signal-to-noise ratios (SNRs). Thus, the RSS-based approaches are ill-suited for LoRa-based localization.
\item As the CSS-modulated LoRa signal is a narrowband signal, it cannot be very sharp in the time domain. Its smooth time-domain waveform will make accurately timing the arrivals of the LoRa signals at the gateway difficult, rendering the time-of-arrival (TOA) approaches futile.
\item To pursue using phase information of the LoRa signal for localization, the frequency variation of CSS makes the estimation of signal phase challenging. In addition, the low-cost LoRa end devices often have considerable frequency biases, especially when deployed in the environments with time-varying conditions (e.g., varying ambient temperature and humidity). The multipath propagation of the LoRa signal may also affect negatively the estimation of the signal phase.\vspace{0.5em}
\end{itemize}

This paper will discuss in detail the challenges faced by a few possible solutions of localizing LoRa end devices. We also provide a formulation of the phase-based approach, which we believe is the most promising solution. The preliminary analysis on the phase-based approach suggests several technical challenges that we need to address in our future work.

The remainder of this paper is organized as follows. Section~\ref{sec:related} reviews related work. Section~\ref{sec:primer} presents the basics of LoRa. Section~\ref{sec:challenges} discusses and analyzes the challenges faced by a few possible approaches. Section~\ref{sec:conclude} concludes this paper.
\section{Related Work}
\label{sec:related}

Using LPWAN radios such as LoRa to implement localization has received increasing research interest. Existing studies on LoRa-based localization can be broadly divided into two classes of time-difference-of-arrivals (TDOA)-based and RSSI-based approaches.

\subsection{TDOA-Based Approaches}
TDOA determines the position of a certain device by the differences among the time instants that the same signal arrives at multiple gateways. Fargas et al.~\cite{fargas2017gps} used timestamps provided by the LoRa gateways to implement a TDOA system. Carvalho et al.~\cite{carvalho2018feasibility} evaluated the feasibility of implementing mobile sensing and tracking applications using LoRa radios. Podevijn et al.~\cite{podevijn2018tdoa} also implemented a TDOA system based on the timestamps extracted from the LoRa gateways to evaluate the tracking performance of a LoRa network. The approaches developed in the above three studies have poor localization performance. This is mainly because that current LoRa hardware implementation and the software stack have a timestamp resolution of one microsecond ($\mu\text{s}$) only. As radio signals travel about 300 meters in free space over a time duration of $1\,\mu\text{s}$, the timing resolution of current off-the-shelf LoRa products is not sufficient for implementing accurate localization~\cite{dongare2017openchirp}. Although there are several signal processing approaches~\cite{bakkali2017kalman, wolf2018improved} such as Kalman filter to improve the robustness against random noises, the timestamp resolution is still the determining factor for the localization performance of TDOA systems.

Rajalakshmi et al.~\cite{nandakumar20183d} designed a multi-band backscatter device based on CSS modulation for three-dimensional localization. The backscatter device is in sub-centimeter form factor and only consumes $93\,\mu\text{W}$ power. The projected lifetime of the device is up to ten years on button cell batteries. The above study achieves meter-level localization accuracy with highly customized LoRa devices. In contrast, we aim to develop localization approaches for off-the-shelf LoRa devices.

\subsection{RSSI-Based Approaches}
Another research thread has focused on RSSI-based localization. Different from TDOA, in an RSSI-based ranging algorithm, a node uses RSSI measurements to estimate its distance from the signal source by using a known signal propagation model that characterizes the relationship between RSSI and distance. Lam et al.~\cite{lam2017lora} described an RSSI-based LoRa localization algorithm in noisy outdoor environments. Accounting for the possible errors in the RSSI measurements, their approach selects the best distance estimate among all estimates. Furthermore, they tried to exclude noisy nodes and select the nodes experiencing lower levels of noises for localization~\cite{lam2018new}. Machine learning approaches have been studied recently. Aernouts et al.~\cite{aernouts2018sigfox} proposed to apply \textit{k}-Nearest-Neighbor (\textit{k}NN) algorithm on the collected RSSI measurements to estimate the object location. Zhe et al.~\cite{he2018enhanced} modeled the RSSI in both indoor and outdoor environments using Gaussian process and then applied maximum likelihood estimation (MLE) to develop a localization approach.

\subsection{Summary}
The existing TDOA-based and RSSI-based approaches have several limitations. Most TDOA approaches based on commodity LoRa devices can only achieve sub-kilometer accuracy. Such accuracy is insufficient for a range of applications, such as unmanned aerial vehicle (UAV) tracking and navigation. Several TDOA approaches require pre-trained models and/or highly customized devices, leading to overhead and high cost in deploying these approaches. The RSSI-based approaches are susceptible to strong signal attenuation and cannot achieve three-dimensional localization well due to the multipath effect.
\section{LoRa Primer}
\label{sec:primer}

LoRa is a physical layer technique that uses a Chirp Spread Spectrum (CSS) modulation and operates in sub-GHz ISM bands (e.g., $868\,\text{MHz}$ in Europe). In LoRa's CSS, a chirp is a finite-time signal with time-varying instantaneous frequency that swaps the whole bandwidth of the communication channel in a linear manner. Given a certain central frequency, denoted by $f_c$, an up chirp's instantaneous frequency increases from $f_c-\frac{BW}{2}$ to $f_c+\frac{BW}{2}$, whereas a down chirp's instantaneous frequency decreases from $f_c+\frac{BW}{2}$ to $f_c-\frac{BW}{2}$. The time duration of a chirp is determined by the \textit{spreading factor} and \textit{bandwidth}, which are denoted by $SF$ and $BW$, respectively. Specifically, the chirp time is given by
\begin{equation}
    t = \frac{2^{SF}}{BW}.
\end{equation}
For the EU868 frequency band, there are six spreading factors, ranging from $7$ to $12$. For example, with $SF = 7$, $BW = 125\,\text{kHz}$, $f_c = 869.75\,\text{MHz}$ and initial phase $\theta = 0$, an up chirp's spectrogram is shown in Fig.~\ref{fig:spetrogram}. Fig.~\ref{fig:3D-preamble} presents the in-phase (I) and quadrature (Q) data of this chirp in the time domain.

\begin{figure}[t]
  \includegraphics[width=\columnwidth]{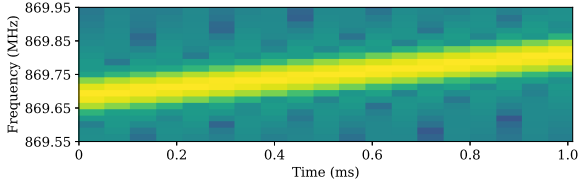}
    \caption{Spectrogram of an up chirp.}
    \label{fig:spetrogram}
  \end{figure}
  
\begin{figure}[t]
  \includegraphics[width=\columnwidth]{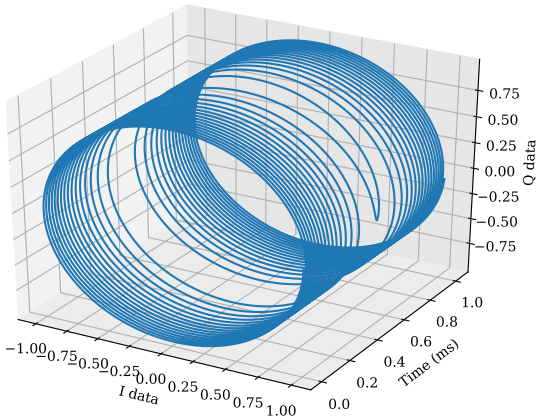}
    \caption{$I$ and $Q$ data $(\theta = 0)$ of an up chirp.}
    \label{fig:3D-preamble}
\end{figure}

The LoRaWAN specification defines three classes of end devices, Class-A, -B, and -C. Class-A devices are the most energy efficient because each communication session must be initialized by an end device rather than the gateway. Therefore, the end devices can follow their own communication schedules with low duty cycles. After an uplink initiated by an end device, there are two optional downlink windows for the gateway to transmit a frame to the end device. A Class-B device will synchronize its internal clock using beacons emitted by the gateway. During the synchronization process, the gateway can schedule the data transmissions and synchronization interval with the end devices. Thus, each Class-B device will use the allocated time slots to transmit uplink frames. Class-C devices will listen to the gateway continuously, which is power-consuming. All off-the-shelf LoRa end devices support the Class-A specification. Class-A operation, though energy efficient, may have a chance of frame collision if two end devices using the same spreading factor transmit at the same time. However, since the duty cycle is generally low, the frame collisions are of low probability. Class-B and -C are free of frame collisions.

We aim to develop an approach that localizes an end device based on its uplink transmissions. The approach will be applicable to all three classes, assuming there are no frame collisions.

\section{Approaches to LoRa-Based Localization and Challenges}
\label{sec:challenges}

Various techniques have been developed for localization using radio signals. In this section, we discuss the challenges of implementing several main techniques for LoRa-based localization.

\subsection{RSSI-Based Approach}
RSSI-based approach builds a path-loss model according to the measured RSSIs at different distances from the signal source. With the model, the distance between the transmitter and the receiver can be estimated. For a static environment, the RSSI-based localization approach can achieve accuracy of hundreds of meters and down to tens of meters. However, radio channels are often subjected to various stochastic and unpredictable factors, especially in the indoor environments. It is generally hard to build an effective path-loss model when the signal travels through walls and floors constructed with diverse materials. Moreover, after penetrating the walls and floors, the LoRa signal strength may go below the noise floor, leading to large errors of the RSSI-based approach.

\begin{figure}
  \includegraphics{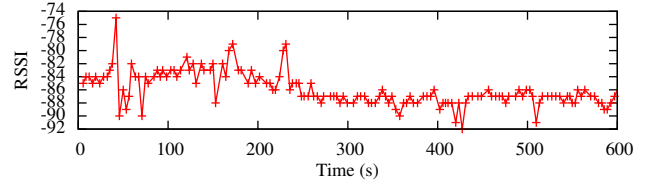}
    \caption{Fluctuations of RSSI within 10 minutes.}
    \label{fig:rssi}
  \end{figure}

According to the datasheet of Semtech SX1276, which is a major commodity LoRa chip on the current market, the RSSI reading is an average value of multiple instantaneous RSSIs. This average RSSI is saved in a register that can be read at any time. We conduct an experiment to show the variation of this average RSSI over time. We set up two nodes as the transmitter and the receiver respectively in an office. Each node consists of an Arduino UNO and an SX1276-based inAir9b LoRa chip. The distance between the transmitter and the receiver is about $10\,\text{m}$. Both of them use $SF = 12$ and $f_c = 859.75\,\text{MHz}$. The transmitter randomly chooses a transmitting interval within 2 to 3 seconds. Fig.~\ref{fig:rssi} shows the average RSSI readings of the receiver over 10 minutes. We can see that the RSSI fluctuates even when both nodes stay still. Such fluctuations challenge the design of the RSSI-based localization approaches.

\subsection{TDOA-Based Approach}
A TDOA-based localization system has multiple base stations with known locations. We also assume that the clocks of the base stations are tightly synchronized. A transmitter transmits a signal that will be captured by all base stations. For each two base stations, the difference between their signal arrival times can be measured. Now, our discussion focuses on two base stations. With the product of the signal propagation speed and the time difference as the real axis, we can get a hyperbola. The location of the transmitter is at the intersection point of the hyperbolas. Thus, the accuracy of timestamping the signal's arrival time is critical to TDOA-based localization approach.

There are mainly three existing approaches to extract the timestamp of an incoming LoRa signal, including reading the timestamp provided by the LoRaWAN gateway, estimating the arrival time in the spectrogram domain and estimating the arrival time in the time domain. All three approaches achieve microsecond level timestamping accuracy only, which is insufficient for meter-level ranging and localization. To achieve meter-level ranging and localization, nanosecond timestamping accuracy will be needed. We now discuss the details of the three possible timestamping methods for LoRa.

\begin{itemize}
\item \textbf{Timestamp given by LoRaWAN gateway:} When a gateway receives a LoRa frame, it will forward it to the LoRaWAN software server with a timestamp. The timestamp is the value of the gateway's internal time counter at the time instant when the LoRa frame was received. However, most commodity LoRaWAN gateways provide time counters with microsecond granularity only \cite{protocol2018}. Thus, the frame arrival timestamping will have microsecond-level accuracy only.
\item \textbf{Spectrogram domain timestamping:} As the up chirp exhibits a clear time-frequency pattern as shown in Fig.~\ref{fig:spetrogram}, a possible approach to locating the signal arrival time is to analyze the spectrogram of the received $I$ and $Q$ data. However, the spectrogram inevitably has reduced time resolution. For instance, the time resolution of the spectrogram in Fig.~\ref{fig:spetrogram} is $1024\,\mu\,\text{s}/20 \approx 50\,\mu\text{s}$, which is beyond the ns-accurate timestamping. Note that the data chirps are decoded by analyzing their spectrograms, because the up and down chirps can be easily differentiated. Differently, preamble onset time estimation imposes more challenges than data decoding.

\begin{figure}
  \subfigure[$\theta = 0$]
  {
    \includegraphics[width=\columnwidth]{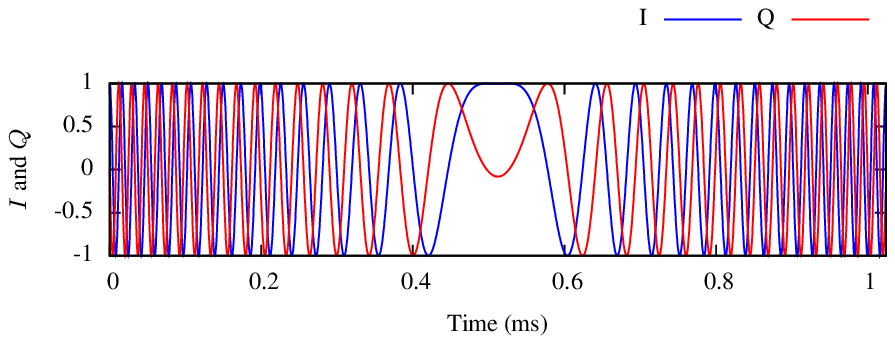}
    \label{fig:0_phase}
  }
  \subfigure[$\theta = \frac{\pi}{2}$]
  {
    \includegraphics[width=\columnwidth]{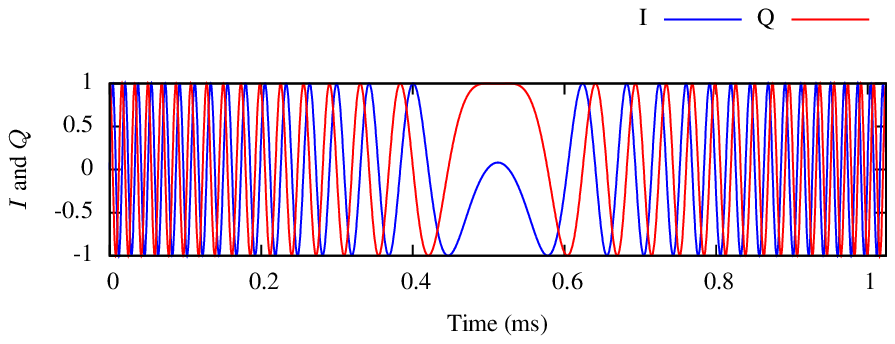}
    \label{fig:180_phase}
  }
  \caption{$I$ and $Q$ data of an up chirp under different $\theta$ settings. Note that $\theta$ is the difference between the initial phases of the carrier signals generated by the transmitter and the receiver.}
  \label{fig:diff_phase}
\end{figure}

\item  \textbf{Time domain timestamping:} Matched filter has been widely adopted for signal arrival time detection. The matched filter requires that the phase of an incoming signal is known to achieve the best detection accuracy. However, as LoRa adopts CSS modulation, the frequency will change with time, it is difficult for a receiver to lock the phase of an incoming signal. Fig.~\ref{fig:diff_phase} shows the ideal $I$ and $Q$ traces of two up chirps when the difference between the initial phases of the carrier signals at the transmitter and the receiver (denoted by $\theta$) are $0$ and $\frac{\pi}{2}$, respectively. We can see that the waveform shapes highly depend on the $\theta$. Since the $\theta$ is generally unknown and hard to estimate for CSS signals, it is difficult to define a template shape for the matched filter to work.
\end{itemize}
\subsection{Phase-Based Approach}

In this section, we will provide a detailed analysis on a possible phase-based approach to LoRa localization.

\begin{table*}
  \begin{minipage}{\textwidth}
    \begin{equation}
      \!\!\!\!\!\!\!\!\!\!\! s_I(t,d) = \frac{\alpha(d)A(t)}{2} \left( \cos \left( 2 \pi \int_{0}^{t-\frac{d}{v}} f(x) \dif x - 2 \pi f_ct + \theta_{\mathrm{Tx}}-\theta_{\mathrm{Rx}} \right) - \cos \left( 2 \pi \int_{0}^{t-\frac{d}{v}} f(x) \dif x + 2 \pi f_ct +\theta_{\mathrm{Tx}}+\theta_{\mathrm{Rx}} \right) \right).
      \label{eq:si}
    \end{equation}
    \begin{equation}
      s_Q(t,d) = \frac{\alpha(d)A(t)}{2} \left( \sin \left( 2 \pi \int_{0}^{t-\frac{d}{v}} f(x)dx - 2 \pi f_ct + \theta_{\mathrm{Tx}} - \theta_{\mathrm{Rx}} \right) + \sin \left( 2 \pi \int_{0}^{t-\frac{d}{v}} f(x)dx)+2 \pi f_ct +\theta_{\mathrm{Tx}}+\theta_{\mathrm{Rx}} \right) \right).
      \label{eq:sq1}
    \end{equation}
    \hrule
  \end{minipage}
\end{table*}

\subsubsection{Modeling of LoRa Chirp Signal Propagation}

As mentioned earlier, a chirp is a finite-time band-pass signal with time-varying frequency. Let $A(t)$ and $f(t)$ denote the instantaneous amplitude and frequency of the chirp signal at the time instant $t$. We formulate the chirp signal emitted by the transmitter from time $t=0$ as $s(t) = A(t) \sin \left( 2\pi \int_{0}^t f(x) \dif x+ \theta_{\mathrm{Tx}}  \right)$, where $\theta_{\mathrm{Tx}} \in [0, 2\pi)$ is the LoRa transmitter's initial phase at $t=0$. The $\theta_{\mathrm{Tx}}$ is usually unknown. As the signal emission starts from time $t=0$, we can make $A(t) = 0$ and $f(t) = 0$ when $t < 0$.
Now we derive the expression of the chirp signal at a position that is $d$ meters from the transmitter. Our analysis ignores the multi-path effect and barriers. The signal at the position, denoted by $s(t, d)$, can be expressed by
\begin{equation*}
  s(t,d)= \alpha(d)A \left(t-\frac{d}{v} \right) \sin \left( 2\pi \int_{0}^{t-\frac{d}{v}} f(x) \dif x+ \theta_{\mathrm{Tx}}  \right),
\end{equation*}
where $\alpha(d) \propto \frac{1}{d^2}$ denotes the attenuation coefficient, $v$ denotes the signal propagation velocity. Note that before the signal front arrives at the position (i.e., $t < \frac{d}{v}$), the signal $s(t, d) = 0$.

\subsubsection{Modeling of LoRa Chirp Signal Reception and a Differential Phase Sampling Technique}

This section models the reception of the LoRa chirp by a gateway implemented by a software-defined radio (SDR).
The SDR generates two unit-amplitude orthogonal carriers $\sin(2\pi f_ct + \theta_{\mathrm{Rx}})$ and $\cos(2 \pi f_c t + \theta_{\mathrm{Rx}})$, where $f_c$ is a specified frequency and $\theta_{\mathrm{Rx}}$ is the initial phase of the two self-generated carriers at $t=0$. The $\theta_{\mathrm{Rx}}$ is usually unknown. Assume the gateway is $d$ meters from the LoRa transmitter. The SDR will mix the received signal with the self-generated carriers, yielding $s_I(t, d) = s(t, d) \cdot \sin(2\pi f_ct + \theta_{\mathrm{Rx}})$ and $s_Q(t, d) = s(t, d) \cdot \cos(2\pi f_ct + \theta_{\mathrm{Rx}})$. The $s_I(t, d)$ and $s_Q(t, d)$ can be further derived as Eq.~(\ref{eq:si}) and Eq.~(\ref{eq:sq1}), respectively. The SDR will apply two internal low-pass filters to remove the high-frequency components in Eq.~(\ref{eq:si}) and Eq.~(\ref{eq:sq1}). Thus, after the low-pass filtering, the $I$ and $Q$ components yielded by the SDR, denoted by $I(t,d)$ and $Q(t,d)$, are given by $I(t,d) = \frac{\alpha(d)A(t)}{2} \cos \Theta(t,d)$, $Q(t,d) = \frac{\alpha(d)A(t)}{2} \sin \Theta(t,d)$, where $\Theta(t,d) =2 \pi \int_{0}^{t-\frac{d}{v}} f(x) \dif x - 2 \pi f_ct + \theta$ and $\theta = \theta_{\mathrm{Tx}}-\theta_{\mathrm{Rx}}$. From the above analysis, the instantaneous phase $\Theta(t, d)$ can be computed by $\Theta(t,d)= \mathrm{atan2} ( Q(t,d), I(t,d) ) + 2k \pi$ where $k \in \mathbb{Z}$. The $k$ rectifies the multi-valued inverse tangent function $\mathrm{atan2}(\cdot, \cdot)\in (-\pi, \pi)$ to an unlimited value domain and ensures that $\Theta(t, d)$ is a continuous function of $t$. Note that under a differential phase sampling (DPS) scheme that will be discussed shortly, the value of $k$ can be easily determined.

We now discuss the DPS scheme. Let $f_s$ denote the sampling rate of the SDR's analog-to-digital converter (ADC). Denote by $\{\Theta[t,d],\Theta[t+\frac{1}{f_s},d],\Theta[t+\frac{2}{f_s},d],...,\Theta[t+\frac{n}{f_s},d]\}$ a sequence of instantaneous phases computed by the sampled $I$ and $Q$ values. We define the DPS sequence $\Delta$ starting from time $t$ with its $i$th element $\Delta[i]$ given by
\begin{align}
 \Delta[i] &= \Theta \left[ t+\frac{i}{f_s},d \right] - \Theta \left[ t+\frac{i+1}{f_s},d \right] \nonumber \\
         &= -2 \pi \int_{t+\frac{i}{f_s}-\frac{d}{v}}^{t+\frac{i+1}{f_s}-\frac{d}{v}} f(x) \dif x+2 \pi f_c\frac{1}{f_s}.
         \label{eq:si2}
\end{align}
The bounds of $\Delta[i]$ are analyzed as follows. The LoRa chirp's instantaneous frequency $f(t) \in \left[f_c-\frac{BW}{2},f_c+\frac{BW}{2} \right]$. Thus, from Eq.~(\ref{eq:si2}), we have $- \frac{BW}{f_s}\pi \le \Delta[i] \leq \frac{BW}{f_s}\pi$.

\begin{prop}
  $f_s > BW$ is a sufficient condition for computing the sequence $\Delta$ unambiguously based on the $I$ and $Q$ traces.
\end{prop}

\begin{proof}
We consider computing a $\Delta$ element based on two consecutive $\Theta$ values $\Theta_1$ and $\Theta_2$, i.e., $\Delta = \Theta_1 - \Theta_2$. With $f_s > BW$, the range of $\Delta$ is $-\pi < \Delta < \pi$. Denote $\Theta_1 = \mathrm{atan2}(Q_1, I_1) + 2k_1 \pi$ and $\Theta_2 = \mathrm{atan2}(Q_2, I_2) + 2k_2\pi$, where $k_1$ is known and $k_2$ is to be determined. Since $-2\pi < \mathrm{atan2}(Q_1, I_1) - \mathrm{atan2}(Q_2, I_2) < 2\pi$, if $k_2 \le k_1 - 2$ or $k_2 \ge k_1 + 2$, $\Delta$ will be out of its range $(-\pi, \pi)$. Thus, there is only three possible cases: i) $k_2 = k_1 - 1$, ii) $k_2 = k_1$, and iii) $k_2 = k_1 + 1$. It can be easily verified that there is a case satisfying $\Delta \in (-\pi, \pi)$ and other two cases must not satisfy $\Delta \in (-\pi, \pi)$. The satisfying case determines the value of $k_2$. Therefore, the elements of $\Delta$ can be determined sequentially without ambiguity.
\end{proof}

For the EU868 frequency band, $BW = 125\,\text{KHz}$. The sampling rate of SDR's ADC (e.g., $20\,\text{Msps}$) is much higher than $BW$. Thus, the condition $f_s > BW$ can be satisfied.


\subsubsection{Distance Difference Estimation and End Device Localization}

The phase-based approach measures the differences among the phases of the gateways' received signals to estimate the differences among the distances between the end device and the gateways. We assume that the clocks of the gateways are tightly synchronized. Note that the clock synchronization between the end device and the gateway is not required. Denote by $d_A$ and $d_B$ the distances from the end device to the gateways $A$ and $B$, respectively. Our following analysis focuses on the estimation of $d_A - d_B$. This approach can be applied to estimate other distance differences, which are then together used by multilateration to localize the end device. The DPS sequence on gateway $A$ is $\Delta_A$, where $\Delta_A[i] = -2 \pi \int_{t+\frac{i}{f_s}-\frac{d_A}{v}}^{t+\frac{i+1}{f_s}-\frac{d_A}{v}} f(x) \dif x+2 \pi f_c\frac{1}{f_s}$. Denote by $F(x)$ the antiderivative of $f(x)$. The DPS sequence on gateway $A$ can be expressed as $\Delta_A[i]= -2 \pi \left( F \left( t-\frac{d_A}{v}+\frac{i+1}{f_s} \right)-F \left( t-\frac{d_A}{v}+\frac{i}{f_s} \right) \right ) + 2 \pi f_c\frac{1}{f_s}$. As the variable in this time sequence is $t-\frac{d_A}{v}$, we can denote the sequence as $\Delta_A = S(t-\frac{d_A}{v})$, where $S(t-\frac{d_A}{v})[i]= -2\pi \left(F\left(t-\frac{d_A}{v}+\frac{i+1}{f_s}\right)-F\left(t-\frac{d_A}{v}+\frac{i}{f_s}\right)\right) + 2 \pi f_c\frac{1}{f_s}$. The DPS sequence on gateway $B$ can be written as $\Delta_B = S(t-\frac{d_B}{v})$. Thus, $\Delta_B = S(t-\frac{d_B}{v})= S(t-\frac{d_A}{v} + (\frac{d_A}{v}-\frac{d_B}{v} ))$. It means that $\Delta_B$ is a time-shifted version of $\Delta_A$, where the time shift is $\frac{d_A}{v}-\frac{d_B}{v}$. Therefore, we can compute the cross-correlation of the two time sequences to estimate time shift. Assuming $K$ to be the peak position of cross-correlation, we have $K \frac{1}{f_s} \approx \frac{d_A}{v}-\frac{d_B}{v}$. The approximation is because $K$ has to be an integer. Thus, the distance difference $d_A - d_B$ can be estimated as $\hat{d}_{AB} = \frac{K \cdot v}{f_s}$.

\subsubsection{Challenges Faced by Phase-Based Approach}
First, the approach needs a high sampling rate $f_s$ to improve the localization accuracy. The resolution of the distance difference estimation is $\frac{v}{f_s}$. It means that if we use $20\,\text{Msps}$ ADC, the resolution is about $15\,\text{m}$. If we use $1\,\text{Gsps}$ ADC, the resolution is about $0.3\,\text{m}$. Future research shall investigate whether upsampling can improve the resolution.

Second, the $I$ and $Q$ traces generally contain noises. After long-distance propagation and barrier penetration, the high noise levels may significantly affect computing $\Theta$ using $\mathrm{atan2}(Q, I)$. DPS scheme robust to noises is to be developed.

Third, our analysis assumes that the end device and the gateways can generate accurate carrier frequency. The frequency bias of the SDR and the end device in generating the carrier signals may affect the localization performance. To measure the impact of temperature change on frequency bias, we conduct a test over 24 hours. We put a transmitter in the corridor of a multistory building that air condition cannot reach and a receiver in a homothermal indoor environment. Fig.~\ref{fig:temperature-bias} shows that the estimated frequency bias changes with temperature. The red circle dot line represents the temperature measured by a temperature sensor in the outdoor environment. The blue square dot line represents the estimation error. Future research shall investigate how the frequency biases affect the distance difference estimation and develop mitigation approaches.

\begin{figure}
  \includegraphics{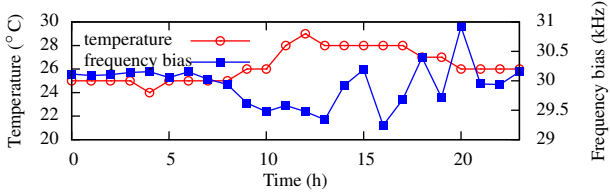}
  \caption{Frequency bias along with temperature over 24 hours in outdoor environment.}
  \label{fig:temperature-bias}
\end{figure}

Finally, the clock synchronization errors among the gateways may effect the accuracy of the phase-based approach. If there exists an unknown synchronization error $\Delta t$ between the gateways $A$ and $B$, the estimated distance difference $\hat{d}_{AB} = (K+\Delta tf_s)\frac{v}{f_s}$. Thus, the resolution is $\max\{1,\Delta tf_s\} \cdot \frac{v}{f_s}$. If we use $1\,\text{Gsps}$ ADC and aim to maintain the resolution at about $0.3\,\text{m}$, we need to synchronize the clocks of the gateways within $1\,\text{ns}$ error. Nanosecond-level clock synchronization is often challenging on commodity platforms.
\section{Conclusion and Future Work}
\label{sec:conclude}
This paper discusses the technical challenges faced by several commonly used techniques in implementing LoRa-based localization. The RSSI-based and TDOA-based approaches do not seem promising due to basic limitations of the LoRa signal (e.g., significant signal attenuation after barrier penetration and smooth time-domain signal waveform) and current commodity LoRa hardware (e.g., low-resolution internal time counter). Our analysis shows that the phase-based approach is more promising but also pinpoint several relevant challenges. Our future work is to implement and evaluate the phase-based approach.


\section{Acknowledgments}

This research is supported by an NTU Start-up Grant.

\balance
\bibliographystyle{abbrv}
\bibliography{sigproc}  

\begin{thebibliography}{10}

\bibitem{protocol2018}
Basic communication protocol between lora gateway and server.
\newblock \url{http://bit.ly/2H304ht}.

\bibitem{lora-alliance}
Lora alliance.
\newblock \url{https://lora-alliance.org/}.

\bibitem{nb-iot}
Narrowband internet of things (nb-iot).
\newblock \url{http://bit.ly/2RwH9Q8}.

\bibitem{sigfox}
Sigfox.
\newblock \url{https://www.sigfox.com/}.

\bibitem{weightless}
Weightless.
\newblock \url{http://www.weightless.org/}.

\bibitem{low-cost}
Wireless sx1276 lora module.
\newblock \url{http://bit.ly/2s3NJiz}.

\bibitem{aernouts2018sigfox}
M.~Aernouts, R.~Berkvens, K.~Van~Vlaenderen, and M.~Weyn.
\newblock Sigfox and lorawan datasets for fingerprint localization in large
  urban and rural areas.
\newblock {\em Data}, 3(2):13, 2018.

\bibitem{bakkali2017kalman}
W.~Bakkali, M.~Kieffer, M.~Lalam, and T.~Lestable.
\newblock Kalman filter-based localization for internet of things lorawan™
  end points.
\newblock In {\em Personal, Indoor, and Mobile Radio Communications (PIMRC),
  2017 IEEE 28th Annual International Symposium on}, pages 1--6. IEEE, 2017.

\bibitem{carvalho2018feasibility}
D.~F. Carvalho, A.~Depari, P.~Ferrari, A.~Flammini, S.~Rinaldi, and E.~Sisinni.
\newblock On the feasibility of mobile sensing and tracking applications based
  on lpwan.
\newblock In {\em Sensors Applications Symposium (SAS), 2018 IEEE}, pages 1--6.
  IEEE, 2018.

\bibitem{dongare2017openchirp}
A.~Dongare, C.~Hesling, K.~Bhatia, A.~Balanuta, R.~L. Pereira, B.~Iannucci, and
  A.~Rowe.
\newblock Openchirp: A low-power wide-area networking architecture.
\newblock In {\em Pervasive Computing and Communications Workshops (PerCom
  Workshops), 2017 IEEE International Conference on}, pages 569--574. IEEE,
  2017.

\bibitem{fargas2017gps}
B.~C. Fargas and M.~N. Petersen.
\newblock Gps-free geolocation using lora in low-power wans.
\newblock In {\em Global Internet of Things Summit (GIoTS), 2017}, pages 1--6.
  IEEE, 2017.

\bibitem{he2018enhanced}
Z.~He, Y.~Li, L.~Pei, and K.~O’Keefe.
\newblock Enhanced gaussian process based localization using a low power wide
  area network.
\newblock {\em IEEE Communications Letters}, 2018.

\bibitem{lam2017lora}
K.-H. Lam, C.-C. Cheung, and W.-C. Lee.
\newblock Lora-based localization systems for noisy outdoor environment.
\newblock In {\em Wireless and Mobile Computing, Networking and Communications
  (WiMob),}, pages 278--284. IEEE, 2017.

\bibitem{lam2018new}
K.-H. Lam, C.-C. Cheung, and W.-C. Lee.
\newblock New rssi-based lora localization algorithms for very noisy outdoor
  environment.
\newblock In {\em 2018 IEEE 42nd Annual Computer Software and Applications
  Conference (COMPSAC)}, pages 794--799. IEEE, 2018.

\bibitem{nandakumar20183d}
R.~Nandakumar, V.~Iyer, and S.~Gollakota.
\newblock 3d localization for sub-centimeter sized devices.
\newblock In {\em Proceedings of the 16th ACM Conference on Embedded Networked
  Sensor Systems}, pages 108--119. ACM, 2018.

\bibitem{podevijn2018tdoa}
N.~Podevijn, D.~Plets, J.~Trogh, L.~Martens, P.~Suanet, K.~Hendrikse, and
  W.~Joseph.
\newblock Tdoa-based outdoor positioning with tracking algorithm in a public
  lora network.
\newblock {\em Wireless Communications and Mobile Computing}, 2018, 2018.

\bibitem{wolf2018improved}
F.~Wolf, C.~Villien, S.~de~Rivaz, F.~Dehmas, and J.-P. Cances.
\newblock Improved multi-channel ranging precision bound for narrowband lpwan
  in multipath scenarios.
\newblock In {\em Wireless Communications and Networking Conference (WCNC),
  2018 IEEE}, pages 1--6. IEEE, 2018.

\end{thebibliography}
\end{document}